\documentclass{article}
\usepackage[square,numbers]{natbib}
\usepackage{amsmath} 
\usepackage{bbold} 
\usepackage{euscript}
\usepackage{amssymb}
\usepackage{epsfig} 
\usepackage{graphicx}
\usepackage{listings} 
\usepackage{moreverb} 
\usepackage{vmargin} 
\usepackage{color}
\usepackage{xspace}

\usepackage{subcaption}
\setpapersize{USletter} 
\allowdisplaybreaks[1]
\setmarginsrb{1in}{0.7in}{1in}{1in}{12pt}{11mm}{0pt}{11mm} 

\usepackage{amsthm}
\theoremstyle{definition}
\newtheorem{definition}{Definition}
\theoremstyle{plain}
\newtheorem{theorem}{Theorem}
\newtheorem{lemma}{Lemma}

\renewcommand{\vec}[1]{\mathbf{#1}}
\newcommand{\defeq}{\triangleq}
\newcommand{\compare}{c_\vec{q}(j,j')}
\newcommand{\pairacc}{PairwiseAccuracy\xspace} 
\newcommand{\pairreg}{pairwise regularization\xspace}
\newcommand{\pairregshort}{Pairwise Regularization\xspace}
\newcommand{\myparagraph}[1]{\textit{#1}. }

\title{Fairness in Recommendation Ranking \\through Pairwise Comparisons}
\author{Alex Beutel, Jilin Chen, Tulsee Doshi, Hai Qian, Li Wei, Yi Wu, Lukasz Heldt,\\Zhe Zhao, Lichan Hong, Ed H. Chi, Cristos Goodrow
\\[1mm] {\small \{alexbeutel, jilinc, tulsee, hqian, liwei, wuyish, heldt, zhezhao, lichan, edchi, cristos\}@google.com }
\\[1mm] {\small Google} }
\date{}

\begin{document}

\maketitle

\begin{abstract}
	Recommender systems are one of the most pervasive applications of machine learning in industry, with many services using them to match users to products or information.  As such it is important to ask: what are the possible fairness risks, how can we quantify them, and how should we address them?

	In this paper we offer a set of novel metrics for evaluating algorithmic fairness concerns in recommender systems.  In particular we show how measuring fairness based on pairwise comparisons from randomized experiments provides a tractable means to reason about fairness in rankings from recommender systems. Building on this metric, we offer a new regularizer to encourage improving this metric during model training and thus improve fairness in the resulting rankings.  We apply this pairwise regularization to a large-scale, production recommender system and show that we are able to significantly improve the system's pairwise fairness.
\end{abstract}

\section{Introduction}
What should we expect of recommender systems?
Recommenders are pivotal in connecting users to relevant content, items
or information throughout the web, but with both users and content producers,
sellers or information providers relying on these systems, it is important that
we understand who is being supported and who is not.  
In this paper we focus on the risk of a recommender system under-ranking groups of items \cite{ekstrand2018exploring,beutel2017beyond,mehrotra2018towards}.  
For example, 
if a social network under-ranked posts by a given demographic group, that could limit the group's visibility on the service.  

\begin{figure}[t]
    \centering
    \begin{subfigure}[b]{0.45\columnwidth}
        \includegraphics[width=\textwidth]{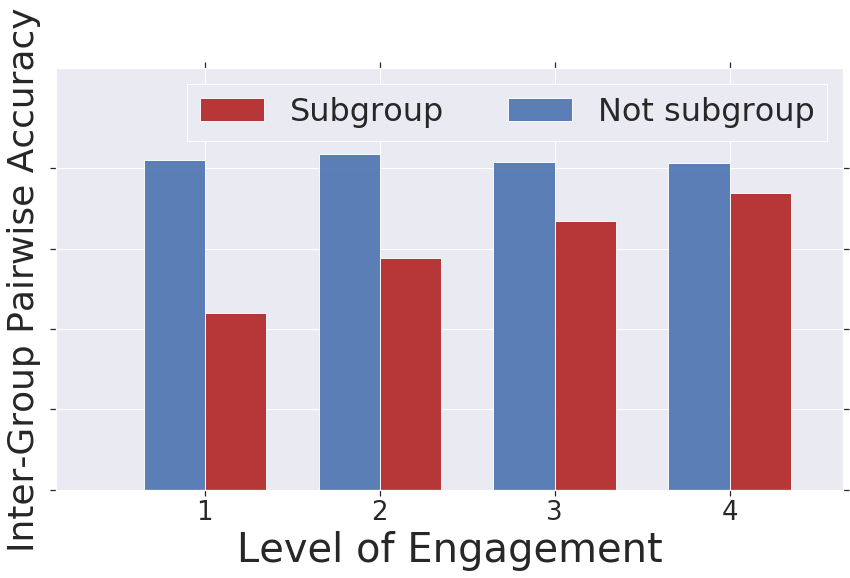}
        \caption{Original}
        \label{fig:prod_fairness_inter}
    \end{subfigure}
    \hspace{2mm}
    \begin{subfigure}[b]{0.45\columnwidth}
        \includegraphics[width=\textwidth]{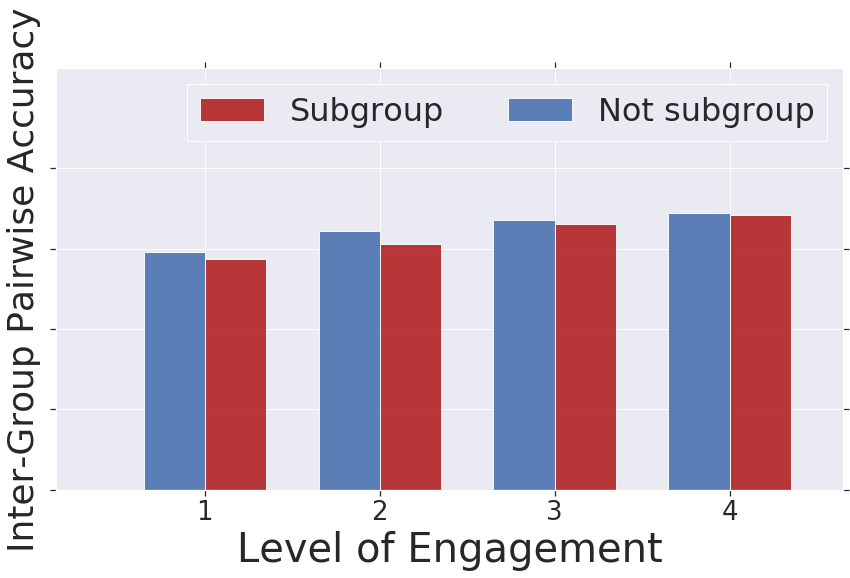}
        \caption{After \pairregshort}
    \end{subfigure}
    \caption{We find significant differences in \emph{inter-group pairwise accuracy}, but using \pairreg we significantly close that gap.}
    \label{fig:fairness_inter}
\end{figure}

While there has been an explosion in fairness metrics for classification
\cite{hardt2016equality,dwork2012fairness,calders2010three,crowson2016assessing} with researchers fleshing out when each metric is appropriate
\cite{kleinberg2016inherent}, there has been far less coalescence of thinking for recommender systems. 
Part of the challenge in studying
fairness in recommender systems is that they are complex.  They often consist of
multiple models \cite{wang2011cascade,he2014practical}, must balance multiple goals
\cite{ma2018modeling,yi2014beyond}, and are difficult to evaluate due to extreme and skewed
sparsity \cite{beutel2017beyond} and numerous dynamics \cite{koren2009collaborative,agarwal2018estimating}.  
All of these issues are hardly resolved in the recommender system community, and present additional challenges in improving recommender fairness.

One challenging split in recommender systems is between treating recommendation as a pointwise prediction problem and applying those predictions for ranked list construction.
Pointwise recommenders make a prediction about user interest for each item and then a ranking of recommendations is determined based on those predictions.  This setup is pervasive in practice \cite{koren2009matrix,mcmahan2013ad,he2014practical,covington2016deep,ma2018modeling}, but significant research goes into bridging the gap between pointwise predictions and ranking construction \cite{adams2011ranking,bello2018seq2slate}.
Fairness falls into a similar dilemma.
Recent research developed fairness metrics centered around pointwise accuracy \cite{beutel2017beyond,yao2017beyond}, but this does not indicate much about the resulting ranking that the user actually sees.
In contract, \cite{zehlike2017fa,DBLP:conf/kdd/SinghJ18,2019arXiv190204056S,biega2018equity} explore what is a fair ranking, but focus on unpersonalized rankings where relevancy is known for all items and in most cases require a post-processing algorithm with item group memberships, which is often not possible in practice \cite{beutel2019putting}.

Further, evaluation of recommender systems is notoriously difficult due to the ever changing dynamics in the system.  What a user was interested in yesterday they may not be interested in tomorrow, and we only know a user's preferences if we recommend them an item.  As a result, metrics are often biased (in the statistical sense) by the
previous recommender system \cite{agarwal2018estimating}, and while a large body of research works to do unbiased offline evaluation \cite{schnabel2016recommendations,schnabel2016unbiased}, this is very difficult due to the large item space, extreme sparsity of feedback, and evolving users and items.  These issues only become more salient when trying to measure recommender system fairness, and even more so when trying to evaluate complete rankings.

We address all of these challenges through a \emph{pairwise recommendation fairness metric}.  Using easy-to-run, randomized experiments we are able to get unbiased estimates of user preferences.  Based on these observed pairwise preferences, we are able to measure the fairness of even a pointwise recommender system, and we show that these metrics directly correspond to ranking performance.  Further, we offer a novel regularization term that we show improves the ultimate ranking fairness of a pointwise recommender, as seen in Figure \ref{fig:fairness_inter}.  We test this on a large-scale recommender system in production and show the practical benefits and trade-offs both theoretically and empirically.
In summary, we make the following contributions:
\begin{itemize}
	\item {\bf Pairwise Fairness:} We propose a set of novel metrics for measuring the fairness of a recommender system based on pairwise comparisons.  We show that this pairwise fairness metric directly corresponds to ranking performance and analyze its relation with pointwise fairness metrics.
	\item {\bf Pairwise Regularization:} We offer a regularization approach to improve the model performance for the given fairness metric that even works with pointwise models.
	\item {\bf Real-world Experiments:} We apply our approach in a large-scale production recommender system and demonstrate that it produces significant improvements in pairwise fairness.
\end{itemize}

\section{Related Work}
This research lies at the intersection of and builds on a myriad of research from the recommender systems community and the machine learning fairness community.

\myparagraph{Recommender Systems}
There is a large research community focused on recommender systems with a wide variety of interests.  Historically, much of the research built on collaborative filtering approaches, with a focus on ratings prediction spurred by the Netflix Prize \cite{koren2009matrix}; another line of work has fleshed out pairwise models of user preferences \cite{joachims2002optimizing,cao2007learning}, particularly for information retrieval in search engines.
As a core component of many industrial application of machine learning, a significant amount of work has been published on production recommenders, such as ad click prediction \cite{mcmahan2013ad,he2014practical} to be used for ranking ads.  These systems often follow cascading patterns with sequences of models being used \cite{wang2011cascade,he2014practical}.
More recently, there has been a strong growth in using state-of-the-art neural network techniques to improve recommender accuracy \cite{he2017neural,covington2016deep,ma2018modeling}.

Building real-world recommenders face a variety of challenges.  Two that relate to the challenges in fairness are the temporal dynamics \cite{koren2009collaborative,wu2017recurrent,hidasi2015session,beutel2018latent} and biased training data \cite{joachims2017unbiased,chen2018top,agarwal2018estimating}.  These issues do not just make training difficult but also evaluation of recommender performance \cite{schnabel2016unbiased}.

\myparagraph{Machine Learning Fairness}
The machine learning fairness community has primarily focused on fairness in classification, with a myriad of definitions being proposed \cite{dwork2012fairness,calders2010three,hardt2016equality,crowson2016assessing}.  Group fairness based definitions, where a model's treatment of two groups of examples is compared, has become the most prevalent structure, but even there researchers have shown tension between different definitions \cite{kleinberg2016inherent,pleiss2017fairness}.  We primarily follow the equality of opportunity intuition from \citet{hardt2016equality}, where we are concerned with differences in accuracy across groups.
Our metric is most closely building on the AUC-based fairness metrics for classification and regression proposed by \citet{dixon2018measuring} and expanded in \cite{manwhitneyfairness} to be framed as different Mann-Whitney U-tests.

\myparagraph{Recommender System Fairness}
There has been a small amount of work on fairness in ranking and recommendation, but with each piece of working taking significantly different perspectives. 
\citet{zehlike2017fa} laid out goals for fair ranking, but do not touch upon recommender systems, where data is far sparser.  Similarly, \citet{DBLP:conf/kdd/SinghJ18}  take a full-ranking view of fairness but are able to apply this to recommender systems through a post-processing algorithm for model predictions; follow-up work \cite{2019arXiv190204056S} moves this into model training.  All of this work \cite{zehlike2017fa,DBLP:conf/kdd/SinghJ18,2019arXiv190204056S,biega2018equity} focuses on an unpersonalized information retrieval setting where relevance labels are known for each item; we focus on personalized recommendations where data sparsity and biases must be handled.  
In contrast, \cite{beutel2017beyond,yao2017beyond} focus on collaborative filtering pointwise accuracy differences across groups but do not connect these metrics to resulting rankings.

More distant is research on statistical parity in recommenders, which argues that in some applications items should be shown at the same rate across groups \cite{zhu2018fairness}.
Diversity \cite{zehlike2017fa,kleinberg2018selection,mehrotra2018towards,stoyanovich2018online}, filter bubbles \cite{bakshy2015exposure}, and feedback loops \cite{feedbackloops}, while related to machine learning fairness, are not the focus of this paper.

\myparagraph{Fairness Optimization}
Many approaches have been proposed to address fairness issues.  Post-processing can provide elegant solutions \cite{hardt2016equality,DBLP:conf/kdd/SinghJ18}, but often requires knowing group memberships for all examples, which is rarely known for demographic data.  Rather, numerous approaches have been developed for optimizing fairness metrics during classifier training, such as constraint-based optimization \cite{agarwal2018reductions,DBLP:conf/nips/GohCGF16}, adversarial learning \cite{zemel2013learning,louizos2015variational,edwards2015censoring,beutel2017data,madras2018learning,DBLP:journals/corr/abs-1801-07593}, and regularization over model predictions \cite{kamishima2011fairness,zafar2015fairness,bechavod2017penalizing,beutel2019putting}.  We build on these regularization approaches for improving the fairness properties of our recommender system.

\section{Pairwise Fairness for Recommendation}
We begin now with a description of our recommender system, the fairness
concerns, and our metrics for them.

\newcommand{\maybeline}{}
\begin{table}
	\centering
	\begin{tabular}{c|l}
		Symbol & Definition\\\hline
		$\vec{q}$ & Query consisting of user and context features \\ \maybeline
		$\mathcal{R}_\vec{q}$ & The set of relevant items for $\vec{q}$ \\\maybeline
		$y, \hat{y}$ & User click feedback and prediction of it \\ \maybeline
		$z, \hat{z}$ & Post-click engagement and prediction of it\\ \maybeline
		$f_\theta(\vec{q},\vec{v})$ & Predictions $(\hat{y},\hat{z})$ for $\vec{q}$ on item $\vec{v}$ \\ \maybeline
		$g(\hat{y}, \hat{z})$ & Monotonic ranking function from predictions \\ \maybeline
		$\compare$ & Comparison between items $j,j'$ based on $g$ and $f$ \\ \maybeline
		$\ell_\vec{q}(j)$ & Position of $j$ in ranked list of $\mathcal{R}_\vec{q}$ \\ \maybeline
		$s_j$ & Binary sensitive attribute for item $j$ \\ \maybeline
		$\mathcal{D}$ & Total dataset of tuples $\{(\vec{q},\vec{v},y,z)\}$ \\ \maybeline
		$\mathcal{P}$ & Dataset of comparisons $\{((\vec{q},\vec{v},y,z),(\vec{q}',\vec{v}',y',z'))\}$
	\end{tabular}
	\caption{Notation used throughout the paper.}
	\label{tab:notation}
\end{table}

\subsection{Recommendation Environment}
\label{sub:formulation}
We consider a production recommender system that is recommending a personalized list of $K$ items to users.  
We consider a cascading recommender \cite{wang2011cascade,he2014practical,covington2016deep}, with a set of retrieval systems
\cite{chen2018top} followed by a ranking system \cite{covington2016deep,ma2018modeling}.
We assume that the retrieval systems return a set $\mathcal{R}$ of $M'$ relevant items from
the total corpus $\mathcal{J}$ of $M$ items, where $M \gg M' \geq K$.  The ranking model must
then score and rank $M'$ items in $\mathcal{R}$ to get a final ranking of $K$
items.  Herein, we focus primarily on the role of the ranker.

Whenever a recommendation is made, the system observes user features $\vec{u}_i$
for user $i$ and a set of context features $\vec{c}$, such as timing or
device information; together we will refer to this as the query $\vec{q} =
(\vec{u},\vec{c})$.  In addition, for each item $j \in \mathcal{J}$ we  observe
feature vector $\vec{v}_j$; this can include sparse representation or learned
embeddings for the item as well as any other properties tied to the item.  
The ranker performs its ranking based on estimates of user feedback, which could
be clicks, ratings \cite{koren2009matrix}, dwell-time on articles \cite{yi2014beyond}, later purchase of items, etc.
For our system, we will estimate if the user clicked on the item $y \in
\{0,1\}$ as well as user engagement after a click on the item $z \in
\mathbb{R}$, such as dwell-time, purchases, or ratings.  As such, our dataset consists of historical examples $\mathcal{D} = \{ (\vec{q},\vec{v},y,z) \}$.
(Note, because $z$ is user engagement after a click, if no click occurs, $z =
0$.)  $\mathcal{D}$ only contains examples that have been recommended
to the user previously.

The ranker is a model $f_\theta$ parameterized by $\theta$;
the model is trained to predict the user engagement $f_\theta(\vec{q},\vec{v}) \rightarrow (\hat{y},\hat{z}) \approx (y,z)$.
Finally, a ranking of the items is produced by a monotonic scoring function $g(\hat{y}, \hat{z})$ and the user is shown the top $K$ items from the relevant items $\mathcal{R}$ ordered by $g$.

\subsection{Motivating Fairness Concerns}
As discussed previously, a wide variety of fairness concerns have been highlighted in the literature. 
In this work, we primarily focus on the risk to groups of items from being under-recommended \cite{ekstrand2018exploring,beutel2017beyond,mehrotra2018towards,2019arXiv190204056S}.  
For example, if a social network under-ranked posts by a given demographic group, that could limit the group's visibility and thus engagement on the service.
If a comment section of a website was personalized, and if a demographic group of users' comments were under-ranked, then that demographic would have less of a voice on the website.
In a more abstract sense, we assume that each item $j$ has sensitive attribute $s_j \in \{0,1 \}$.  We will work to measure if items from one group are systematically under-ranked.

Although not our primary focus, these issues could align with user group concerns if a group of items is preferred by a particular user group.  
This framework could be explicitly extended to incorporate user groups as well.  If each user has a sensitive attribute, we can compute \emph{all} of the following metrics over each user group and compare performance across the groups.  For example, if we are concerned that a social network is under-ranking items about a particular topic to a particular demographic, we could compare the degree of under-ranking of that topic's content across demographic groups.

\subsection{Pairwise Fairness Metric}
\label{sub:metrics}
While the above fairness goals may seem worthwhile, we must make precise what it means for an item to be ``under-ranked.''  
Here, we draw on the intuition of \citet{hardt2016equality} for equality of
odds, where the fairness of a
classifier is quantified by comparing either its false positive rate and/or false negative rate.  Stated differently, given an item's label is positive, what is the probability the classifier will predict it to be positive.  This works well in classification because the model's prediction can be compared to a predefined threshold.

In recommender systems, it is less clear what a positive prediction is, even if we restrict our analysis to clicks ($y$) and ignore engagement ($z$).  For
example, if an item is clicked, $y=1$, and the predicted probability of click
is $\hat{y} = 0.6$, is this a positive prediction?  It can be perceived as an
under-prediction of $0.4$, but it may still be the top-ranked item if all other
items have a predicted $\hat{y} < 0.6$.  As such, understanding
errors in the pointwise predictions requires comparing the predictions of items
for the same query.

We begin with defining a pairwise accuracy: what is the probability
that a clicked item is ranked above another relevant unclicked item, for the same query:
\begin{align}
	{\rm \pairacc} \defeq P( g(f_\theta(\vec{q},\vec{v}_j)) > g(f_\theta(\vec{q},\vec{v}_{j'})) | y_{\vec{q},j} > y_{\vec{q},j'}, j,j' \in \mathcal{R}_\vec{q})
\end{align}
With this definition, we have a sense of how often the ranking system ranks the clicked item well.
For succinctness, we will use $c_\vec{q}(j,j') \defeq \mathbb{1}[g(f_\theta(\vec{q},\vec{v}_j)) > g(f_\theta(\vec{q},\vec{v}_{j'}))]$ to represent the comparison between the predictions for item $j$ and $j'$ on query $\vec{q}$; we will hide the term $j,j' \in \mathcal{R}_\vec{q}$, but we only consider comparisons among relevant items for all following definitions.

As with much of the rest of fairness research, we are concerned with the
relative performance across groups, not the absolute performance.
As such, we can compare:
\begin{align*}
	P( \compare | y_{\vec{q},j} > y_{\vec{q},j'}, s_j = 0)
	 = P( \compare | y_{\vec{q},j} > y_{\vec{q},j'}, s_j = 1)
\end{align*}
That is, is the \pairacc for items from one group $S = 0$ higher or lower than the \pairacc for items from the other group $S=1$\footnote{We focus on $s_j$ rather than $s_{j'}$ as item $j$ is the clicked item and we are concerned with ranking the clicked item well.  We will incorporate $s_{j'}$ in later metrics.}.

While this is an intuitive metric, it is problematic in that it ignores entirely user engagement $z$ and thus possibly runs the risk of promoting clickbait that users don't ultimately value.  As such, we can follow the approach of conditioning on other dependent signals as in \cite{ritov2017conditional,beutel2019putting}.
\begin{definition}[Pairwise Fairness]
	A model $f_\theta$ with ranking formula $g$ is considered to obey pairwise
	fairness if the likelihood of a clicked item being ranked above another
	relevant unclicked item is the same across both groups, conditioned on the items have been engaged with the same amount: 
\begin{align}
	P( \compare | y_{\vec{q},j} > y_{\vec{q},j'}, s_j = 0, z_{\vec{q},j} = \tilde{z})  = P( \compare | y_{\vec{q},j} > y_{\vec{q},j'}, s_j = 1, z_{\vec{q},j} = \tilde{z}), \;\; \forall \tilde{z}.
\end{align}
\end{definition}
\noindent This definition gives us an aggregate notion of ranker accuracy for items from each group.

While this is valuable, it does not distinguish between types of mis-orderings.
This can be problematic in systematically under-exposing items from one group
\cite{DBLP:conf/kdd/SinghJ18}.  For illustration, consider the following two examples where in both cases there are three items from each group $\{A_j\}_{j=1}^3 \cup \{B_j\}_{j=1}^{3}$ and in the first case assume that $A_1$ is clicked and in the second case assume that $B_1$ is clicked.  If in the first case the system gives a ranking $[A_2, A_3, B_1, A_1, B_2, B_3]$ and in the second case the systems gives $[A_1, A_2, A_3, B_1, B_2, B_3]$, we see that the overall pairwise accuracy is the same in both cases, $\frac{2}{5}$, but in the second case, even when an item from group $B$ was of interest (clicked), all group $B$ items ranked below group $A$ items.  Both are problematic in ranking the clicked item low, but the second is more problematic in systematically preferring one group to the other, independent of user preferences.

To deal with this, we can split the above pairwise fairness definition into two
separate criteria: pairwise accuracy between items in the same group and
pairwise accuracy between items from different groups; we will refer to these
metrics as \emph{intra-group pairwise accuracy} and \emph{inter-group pairwise accuracy}, respectively:
\begin{align}
\mbox{Intra-Group Acc.} &\defeq P(\compare | y_{\vec{q},j} > y_{\vec{q},j'}, s_j = s_{j'}, z_{\vec{q},j} = \tilde{z})
\\\mbox{Inter-Group Acc.}  &\defeq  P(\compare | y_{\vec{q},j} > y_{\vec{q},j'}, s_j \neq s_{j'}, z_{\vec{q},j} = \tilde{z})
\end{align}
From these we can define 
\emph{Intra-Group Pairswise Fairness} and \emph{Inter-Group Pairwise
Fairness} criteria.

\begin{definition}[Intra-Group Pairwise Fairness]
	A model $f_\theta$ with ranking formula $g$ is considered to obey intra-group pairwise
	fairness if the likelihood of a clicked item being ranked above another
	relevant unclicked item from the same group is the same independent of group,
	conditioned on the items have been engaged with the same amount: 
\begin{align}
	P( \compare | y_{\vec{q},j} > y_{\vec{q},j'}, s_j = s_{j'} = 0, z_{\vec{q},j} = \tilde{z})  = P( \compare | y_{\vec{q},j} > y_{\vec{q},j'}, s_j =  s_{j'} = 1, z_{\vec{q},j} = \tilde{z}), \forall \tilde{z}.
\end{align}
\end{definition}

\begin{definition}[Inter-Group Pairwise Fairness]
	A model $f_\theta$ with ranking formula $g$ is considered to obey inter-group pairwise
	fairness if the likelihood of a clicked item being ranked above another
	relevant unclicked item from the opposite group is the same independent of group,
	conditioned on the items have been engaged with the same amount: 
\begin{align}
	P( \compare | y_{\vec{q},j} > y_{\vec{q},j'}, s_j = 0, s_{j'} = 1, z_{\vec{q},j} = \tilde{z}) = P( \compare | y_{\vec{q},j} > y_{\vec{q},j'}, s_j = 1, s_{j'} = 0, z_{\vec{q},j} = \tilde{z}), \forall \tilde{z}.
\end{align}
\end{definition}
The Intra-Group Pairwise Fairness to some degree acts similarly to the overall
Pairwise Fairness notion as it indicates the ability of the recommender system
to rank well the item of interest to the user.  The Inter-Group Pairwise
Fairness gives us further insight into whether mistakes in ranking are at the
cost of the group as a whole.

We can see this more clearly by decomposing the overall pairwise accuracy as follows:
\begin{align}
\label{eq:marginal}
&P( \compare | y_{\vec{q},j} > y_{\vec{q},j'}, s_j = \tilde{s}, z_{\vec{q},j} = \tilde{z}, j,j' \in \mathcal{R})
\\&\hspace{6mm} = \sum_{\tilde{s}' = \{\tilde{s}, 1-\tilde{s}\}} 
P( s_{j'} = \tilde{s}' | y_{\vec{q},j} > y_{\vec{q},j'}, s_j = \tilde{s}, z_{\vec{q},j} = \tilde{z}, j,j' \in \mathcal{R}) \nonumber
\\&\hspace{21mm}
\times P( \compare | y_{\vec{q},j} > y_{\vec{q},j'}, s_j = \tilde{s}, s_{j'} = \tilde{s}', z_{\vec{q},j} = \tilde{z}, j,j' \in \mathcal{R}) \nonumber
\end{align}
That is, we find we can break up the pairwise comparisons into two sets, intra-group and inter-group comparisons, and that the overall pairwise accuracy is a weighted sum of the inter-group accuracy and intra-group accuracy, where the weights are determined by the probability of seeing a pair of that form (inter-group or intra-group) with the corresponding click and engagement.
Together, these metrics give us a better sense of the fairness of the recommender system.

\subsection{Measurement}
\label{sub:measure}
While the above definitions offer a goal of how we would like a recommender
system to perform, measuring the degree to which a recommender system meets
these goals presents unique challenges.  As discussed in the introduction,
users and items in recommender systems are highly dynamic, and we typically
only observe user feedback on previously recommended items, which makes metrics
vulnerable to bias in the previous recommender system.

However, for all three fairness definitions given above, we would like to have
unbiased estimates of user preferences between pairs of items.  In order to do
this we run \emph{randomized experiments} over a small percentage of queries
to the recommender system.  The experimental description below is all assumed
to operate over the subset of queries in the experimental slice.

For the experimental queries we will show the user a pair of items in positions
two and three of the recommended slate; this prevents any position bias
\cite{agarwal2018estimating} where items that the recommender system
ranks low are less likely to be clicked than items ranked high, irrespective of the particular items.
Because the definitions above are all over arbitrary pairs of items from the set of relevant items for the given query, for each query two items are chosen at random from $\mathcal{R}_\vec{q}$ and their ordering in positions two and three is also randomized.

Among all queries in the experimental slice only a small fraction will have
clicks on one of the items in the randomized item pair.  Whenever an item in the randomized item pair is clicked, we record the query, pair, which item was clicked, and the subsequent engagement $z$.  With this, we can compute all of the probabilities in the fairness definitions above. In practice we discretize $z$ into buckets for easier comparison.

Note, as can be seen through this experiment, we cannot know the engagement we would have observed if the unclicked item had been clicked.  This motivates our current metric design of conditioning on $z$ rather than estimating the accuracy of $\hat{z}$, since we can only know $\hat{z}$ for one item in the pair.

\paragraph{Discussion}
These metrics connect the performance of the ranking model to the end fairness properties of the resulting ranking.  One underlying assumption is that the retrieval system that determines the set of relevant items, $\mathcal{R}_\vec{q}$, is in some sense ``fair.''  We believe further research is needed to understand both what does it mean for a retrieval system to be ``fair'' and how any degree of bias in the retrieval system propagates through the ranking system to effect the end ranking experience.

\section{Theoretical Analysis}
While hopefully the above definitions are clear and well motivated, we find
that upon further inspection they obey a number of fascinating properties.

\subsection{Ranking Interpretation}
While the metrics have thus far primarily been described similar to pairwise
accuracy, they can be interpreted through the lens of ranking.  
That is, the recommender system sorts $\mathcal{R}_\vec{q}$
according to $g$ and $f_\theta$.  We will use
$\ell(j)_\vec{q}$ to denote the position of item $j \in \mathcal{R}_\vec{q}$ in the sorted list of items:
\begin{align}
	\ell(j)_\vec{q} = |\mathcal{R}_\vec{q}| - \sum_{j' \in \mathcal{R}_\vec{q} \setminus j} \mathbb{1}[g(f_\theta(\vec{q},\vec{v}_j)) > g(f_\theta(\vec{q},\vec{v}_{j'}))]
\end{align}
From this perspective, we find that we can connect pairwise fairness to fairness with respect to the ranked position:
\begin{theorem}
	If a recommender system achieves pairwise fairness then the expected
	position of a clicked item with $z$ engagement is the same across groups.
\end{theorem}
\begin{proof}
	This falls out of the definition of pairwise accuracy and pairwise fairness:
	\begin{align*}
		&\mathbb{E}_{\vec{q} | y_{\vec{q},j} = 1, s_j = 0, z_{\vec{q},j} = \tilde{z}}[\ell(j)_\vec{q}] 
		\\&\hspace{10mm}= \mathbb{E}_{\vec{q}}[|\mathcal{R}_\vec{q}| (1 - {\rm \pairacc}_\vec{q}(j))]
		\\&\hspace{10mm}= \mathbb{E}_{\vec{q}}[|\mathcal{R}_\vec{q}| 
	(1 - P( \compare | y_{\vec{q},j} > y_{\vec{q},j'}, s_j = 0, z_{\vec{q},j} = \tilde{z}))]
		\\ &\hspace{10mm}= \mathbb{E}_{\vec{q}}[|\mathcal{R}_\vec{q}| 
	(1 - P( \compare | y_{\vec{q},j} > y_{\vec{q},j'}, s_j = 1, z_{\vec{q},j} = \tilde{z}))]
		\\ &\hspace{10mm}= \mathbb{E}_{\vec{q} | y_{\vec{q},j} = 1, s_j = 1, z_{\vec{q},j} = \tilde{z}}[\ell(j)_\vec{q}]
	\end{align*}
\end{proof}
As such, we see that we can interpret pairwise recommender fairness as equivalent to the notion
that the position of a clicked and engaged with item should not depend on the group membership on average, aligning with position bias highlighted in \cite{biega2018equity,2019arXiv190204056S}.  (This analysis is similar to probabilistic interpretations in traditional pairwise IR \cite{cao2007learning}, but now in the context of recommender system fairness.)

The inter-group and intra-group pairwise accuracies also connect to the rank position of the clicked item.  That is, we can decompose:
\begin{align*}
	\ell(j)_\vec{q} &= |\mathcal{R}_\vec{q}| - \sum_{j' \in \mathcal{R}_\vec{q} \setminus j} \mathbb{1}[g(f_\theta(\vec{q},\vec{v}_j)) > g(f_\theta(\vec{q},\vec{v}_{j'}))]
	\\ &=|\mathcal{R}_\vec{q}| -  \left(\sum_{j' \in \mathcal{R}_\vec{q} \setminus j | s_j = s_{j'}} c_\vec{q}(j,j')
	+ \sum_{j' \in \mathcal{R}_\vec{q} \setminus j | s_j \neq s_{j'}} c_\vec{q}(j,j')\right)
	\\ &= |\mathcal{R}_\vec{q}| (1 - P(s_{j'} = s_j | j' \in \mathcal{R}_\vec{q})P(c_\vec{q}(j,j') | s_{j'} = s_j, j' \in \mathcal{R}_\vec{q}) 
	\\&\hspace{14mm} - P(s_{j'} \neq s_j | j' \in \mathcal{R}_\vec{q})P(c_\vec{q}(j,j') | s_{j'} \neq s_j, j' \in \mathcal{R}_\vec{q}) )
	\\ &= |\mathcal{R}_\vec{q}| - |\{j' | j' \in \mathcal{R}_\vec{q}, s_j = s_{j'} \}|P(c_\vec{q}(j,j') | s_{j'} = s_j, j' \in \mathcal{R}_\vec{q}) 
	\\&\hspace{11mm} - |\{j' | j' \in \mathcal{R}_\vec{q}, s_j \neq s_{j'} \}|P(c_\vec{q}(j,j') | s_{j'} \neq s_j, j' \in \mathcal{R}_\vec{q})
\end{align*}
Here too, we see that the overall ranked position can be decomposed into the position within the ranked list from the same group and position among the ranked list from the other group.  However, because of the possibly varying distributions of the number of comparisons of each type, we believe it makes sense to focus on each of these terms as probabilities.

\subsection{Relation to Pointwise Metrics}
\label{sub:theory_relations}
While our pairwise fairness metric aligns with previously stated goals of fair
ranking, we find that it lies in tension with traditional pointwise metrics.
For example, recommender systems are often evaluated in terms of calibration or
RMSE \cite{koren2009matrix}, and these metrics have been espoused as important fairness
metrics in classification \cite{crowson2016assessing} and in recommendation
\cite{beutel2017beyond,yao2017beyond}.
We show here that these pointwise metrics are insufficient for guaranteeing
pairwise fairness.

For the following proofs we consider a simplified case where $z = 0$ and
$g(y,z) \defeq y$.  This can be thought of ranking items by predicted click
through rate (pCTR). For each group $s$ we denote its average label $y$ by
$\bar{y}_s \defeq \mathbb{E}_{\vec{q},j}[y_{\vec{q},j} | s_j = s]$.

\paragraph{Calibration}
We begin with examining the relationship between calibration and pairwise fairness.
A pCTR model $f(x)$ for labels $y$ is considered calibrated if and only if:
\begin{align}
	E[y | f(x) = \tilde{y}] = \tilde{y}
\end{align}
That is, among examples receiving a particular prediction, the average label
for those examples needs to be equal to the predicted value.  In the context of
fairness, this would be evaluated over examples from one group.

\begin{lemma}
A calibrated model is insufficient for guaranteeing pairwise ranking fairness.
\label{lemma:calibration}
\end{lemma}
\begin{proof}
	In order to prove this, we offer an example of a calibrated model that does not obey pairwise ranking fairness.
	Let's assume that we learn a model that predicts $f(\vec{q},\vec{v}_j)
	\defeq \bar{y}_{s_j}$ for all examples with an item from group $s$.  
	This model is by definition calibrated per group.  
	
	If we have two groups, $\tilde{s}$ and $\tilde{s}'$, where $\bar{y}_{\tilde{s}} >
	\bar{y}_{\tilde{s}'}$ then $P(\compare) = 1$ for all items $j$ with $s_j = \tilde{s}$ and all items $j'$ with $s_{j'} = \tilde{s}'$.  As such, 
	$P( \compare | y_j > y_{j'}, s_j = \tilde{s}, s_{j'} = \tilde{s}') = 1$ and $P( \compare | y_{j'} > y_{j}, s_j = \tilde{s}, s_{j'} = \tilde{s}') = 0$.
	From this it is clear that inter-group fairness does not hold. 
	We assume ties are split randomly, giving us $P( \compare | y_j > y_{j'}, s_j = s_{j'}) = 0.5$.  
	Based on Eq. \eqref{eq:marginal}, we find that as long as $P(s_j \neq s_{j'} | j,j' \in \mathcal{R}_\vec{q}, y_j > y_{j'}) > 0$ then overall pairwise fairness does not hold.
\end{proof}
This problem is highly similar to the issue pointed out by \cite{DBLP:conf/kdd/SinghJ18} with respect to ranking and exposure, and we see here holds true even among pairwise comparisons.

\paragraph{Squared Error}
Another pervasive metric in recommender systems is mean squared error
(MSE) \cite{koren2009matrix}.  This metric, and modifications of it, have been
proposed for evaluating the fairness of collaborative filtering systems
\cite{beutel2017beyond,yao2017beyond}.  While this may be worthwhile to
encourage accuracy across groups,we find that it too is insufficient for
guaranteeing pairwise fairness.

\begin{lemma}
Equal MSE across groups is insufficient for guaranteeing pairwise ranking fairness.
\end{lemma}
\begin{proof}
	As above we will demonstrate an example of a model that achieves equal MSE across groups but does not obey pairwise fairness.
	Again, let us assume that we learn a model that predicts $f(\vec{q},\vec{v}_j) \defeq \bar{y}_{s_j}$.
	We assume we have two groups, $\tilde{s}$ and $\tilde{s}'$, where 
$\bar{y}_{\tilde{s}} = 1 - \bar{y}_{\tilde{s}'}$ and
$\bar{y}_{\tilde{s}} >     \bar{y}_{\tilde{s}'}$.  
We see that:
\begin{align*}
	{\rm MSE}_s &= \bar{y}_s (1 - \bar{y}_s)^2 + (1 - \bar{y}_s) (0 - \bar{y}_s)^2
	= \bar{y}_s - \bar{y}_s^2
\end{align*}
Through simply substituting in $\bar{y}_{\tilde{s}} = 1 - \bar{y}_{\tilde{s}'}$ in the definition above we can see that ${\rm MSE}_{\tilde{s}} = {\rm MSE}_{\tilde{s}'}$. 

Just as in the proof above, because
$\bar{y}_{\tilde{s}} >
	\bar{y}_{\tilde{s}'}$, we find $P(\compare) = 1$ for all items $j$ with $s_j = \tilde{s}$ and all items $j'$ with $s_{j'} = \tilde{s}'$.  As such, 
	$P( \compare | y_j > y_{j'}, s_j = \tilde{s}, s_{j'} = \tilde{s}') = 1$ and  $P( \compare | y_{j'} > y_{j}, s_j = \tilde{s}, s_{j'} = \tilde{s}') = 0$.  Again, it is clear that inter-group fairness does not hold.
	If we split ties randomly, such that $P( \compare | y_j > y_{j'}, s_j = s_{j'}) = 0.5$, and as long as $P(s_j \neq s_{j'} | j,j' \in \mathcal{R}_\vec{q}, y_j > y_{j'}) > 0$, then we again find that overall pairwise fairness does not hold either.
\end{proof}
As a result, while matching MSE across groups is intuitively valuable for fairness, it is insufficient for making any claims about the end ranking.  
Looking at the example in the proof, it is clear that this is in part due to the fact that MSE does not distinguish between over- and under-prediction.  
However, even taking that into account, MSE ignores the relative ranking and thus it is hard to determine what an improvement of $\epsilon$ in MSE means for ranking accuracy.

\section{Pairwise Regularization to Improve Fairness}
\label{sec:modeling}
With an understanding of our fairness goals, we now
ask: how can we learn a recommender system that achieves these fairness
properties?  As discussed previously, most production recommenders are
pointwise recommenders trained to predict $y$ and $z$, so we would like a
modeling approach that does not require throwing out existing techniques.

To encourage fairness during training, we build on the regularization approach first proposed by \citet{zafar2015fairness} and expanded upon by \citet{beutel2019putting}.  In particular, \citet{beutel2019putting}
optimized for equality of opportunity in classification
\cite{hardt2016equality} by minimizing the correlation between the group
membership and the model's predictions among examples with the same label.
In our setting, we are concerned with the relative ordering, so we must modify this objective.

We assume that our model is trained with a loss $L(f_\theta(\vec{q},\vec{v}), (y,z))$; for example, if squared error were used then 
$L(f_\theta(\vec{q},\vec{v}), (y,z)) \defeq (\hat{y} - y)^2 + (\hat{z} - z)^2$.
Further, we assume that we know $g(\hat{y},\hat{z})$ and that it is
differentiable.  Given this, we train our model $f_\theta$ with the following objective:
\begin{align}
	\min_\theta & \left( \sum_{ (\vec{q},j, y, z) \in \mathcal{D} } L(f_\theta(\vec{q},\vec{v}_j), (y,z)) \right) + |{\rm Corr}_\mathcal{P}(A,B)|
\end{align}
Here, $\mathcal{D}$ is our original training data and $\mathcal{P}$ is the
experimental data from Section \ref{sub:measure} consisting of pairs of
tuples $ ((\vec{q},j, y, z), (\vec{q}',j', y', z') )$.  The second term, the
absolute correlation is computed as the correlation between two terms, $A$
and $B$, both random variables over pairs from $\mathcal{P}$:
\begin{align}
	A &=  (g(f_\theta(\vec{q},\vec{v}_j)) - g(f_\theta(\vec{q}',\vec{v}_{j'}))) (y - y')
	\\B &=  (s_j - s_{j'})(y - y')
\end{align}
That is, the pairwise regularizer calculates the correlation between the
residual between the clicked and unclicked item and the group membership of the
clicked item.  As a result, the model is penalized if the its ability to predict which item was clicked is better for one group than the other.

To make sure there is sufficient data for a meaningful calculation, we
rebalance $\mathcal{P}$ to have approximately half of the data with the clicked item belonging to group $s = 0$ and the other half with the clicked item belonging to group $s = 1$.
Further data restrictions can be applied for alternative goals. If we are concerned with intra-group pairwise fairness, we can restrict $\mathcal{P}$ to the set of pairs where $s_j \neq s_{j'}$, and if we are
concerned with a large differences in engagement $z$, we can create buckets of $\mathcal{P}_z$ where all pairs in the set resulted in engagement $z$.
This approach is general enough to be used with pointwise recommenders as well as pairwise recommenders.

As in \cite{beutel2019putting}, this approach does not provably achieve
pairwise fairness but we follow it due to its strong empirical
performance and ease of use, crucial for production
applications.

\section{Experiments}
To understand our pairwise fairness metric and our proposed modeling improvements, we study the performance of a large-scale, production recommender system.  We offer analysis of the state-of-the-art production model's performance as well as how our modeling changes effect the system.

\subsection{Experimental Setup}
As described in Section \ref{sub:formulation}, we study a cascading recommender system where multiple retrieval systems return the set of relevant items for a given query, followed by a ranking model.  Here we evaluate the ranking model's performance.  The ranking model is a multi-layer neural network trained in a pointwise-fashion with multiple heads to predict the probability of a click, $y$, and a set of user engagement signals after the click, which we refer to in the aggregate by $z$; this is a similar setup to \cite{he2017neural,ma2018modeling}.  This model is continuously trained on a dataset of interactions with previous recommendations.

We study the performance of the ranker with respect to a sensitive subgroup of items, comparing the performance of this subgroup to the rest of the data, denoted by ``not subgroup.''  The subgroup represents approximately 0.2\%  of all items, but it is a subgroup that we feel is important to analyze for recommendation fairness.  As mentioned previously, we only know the group membership for a small percent of items;  this prevents using serving-time approaches to improve the pairwise fairness metrics.
Following the description in Section \ref{sub:measure}, we gather a dataset $\mathcal{P}$ of random pairs of relevant items shown to the user and recorded when the user clicks on one of the items.  We use a random half of this dataset for the \pairreg and the other half for evaluating the model.

We compare two versions of the model: (1) the production model trained without any attention to fairness, (2) a test model, trained with the same architecture but with the \pairreg to optimize for inter-group pairwise fairness.  
(As this is a live system with both data and training dynamics, we present a model chosen at random from a set of test models.)
As we will see below, we focus on inter-group pairwise fairness as this is the area we find needing more improvement.

Due to the sensitive nature, we cannot report absolute accuracy measures.  Rather, we report the relative performance between the subgroup and the rest of the data.  That is, we aggregate the pairwise accuracy measures across engagement levels through a simple average; and we report the relative ratio of the average accuracy for the ``not subgroup'' divided by the average accuracy for the subgroup.  All plots group engagement $z$ into four levels and maintain the same y-axis scaling so that relative comparisons can be made across them.

\subsection{Baseline Performance}
We begin with an analysis of the production system's performance.  As discussed in Section \ref{sub:metrics}, we analyze the system's performance in terms of: (1) pairwise fairness, (2) intra-group pairwise fairness, and (3) inter-group pairwise fairness.

As the overall pairwise fairness in Figure \ref{fig:prod_fairness_overall} shows,
the production system under-ranks items from the subgroup when the subsequent level of engagement is low, but interestingly slightly over-ranks items from the subgroup when the subsequent level of engagement is high.  In total, we find that the non-subgroup items have an 8.3\% advantage overall\footnote{That is, the  pairwise accuracy of ``not subgroup'' divided by the pairwise accuracy of ``subgroup'' is 1.083.}.

Second, we examine the performance within each group -- the intra-group pairwise accuracy.  As can be seen in Figure \ref{fig:prod_fairness_intra}, across all levels of engagement the model has more difficulty selecting the clicked item when comparing subgroup items than when comparing non-subgroup items.  In total, this puts the non-subgroup items at a 14.9\% advantage in intra-group pairwise fairness.  We have found that this is in part due to the subgroup being small while there is far more diversity among the non-subgroup items, making comparisons easier.  When further filtering the subgroup comparisons to remove highly-similar item comparisons, we find no meaningful difference in performance between the subgroup and the non-subgroup.

While both of the above results suggest some deficiencies, we find that the story is significantly more dramatic when looking at the inter-group pairwise accuracy.  As seen in Figure \ref{fig:prod_fairness_inter}, across all levels of engagement we find that the subgroup items are significantly under-ranked relative to the non-subgroup items.  Overall, we find that the non-subgroup items have a 35.6\% advantage.  Further, we see that the pairwise accuracy for non-subgroup items in inter-group pairs is notably higher than in intra-group pairs, suggesting that the model is taking advantage of this difference in items.  
This suggests that subgroup items, even when of interest to the user, are ranked under non-subgroup items.  Because of the implication on subgroup experience and the more dramatic nature of the results, we focus herein on improving the inter-group pairwise fairness.

\begin{figure}[t]
    \centering
    \begin{subfigure}[b]{0.45\columnwidth}
        \includegraphics[width=\textwidth]{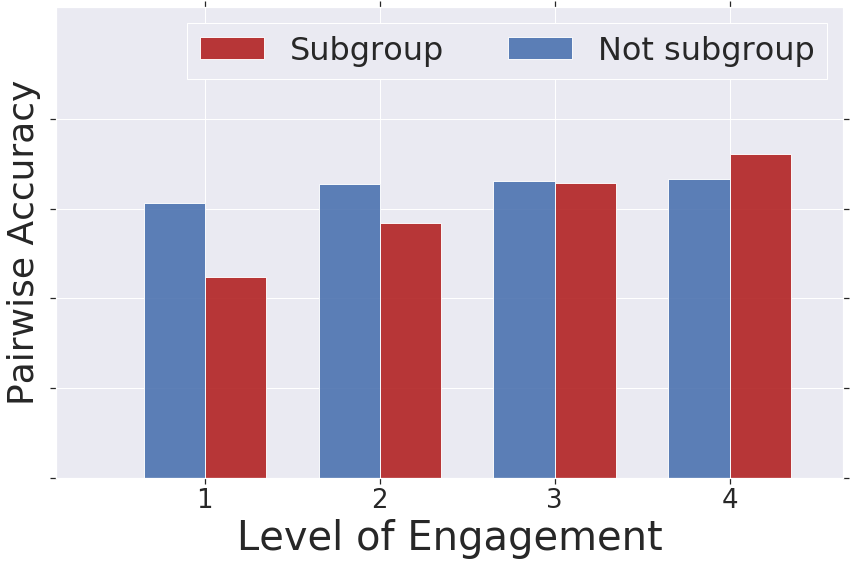}
        \caption{Original}
        \label{fig:prod_fairness_overall}
    \end{subfigure}
    \hspace{2mm}
    \begin{subfigure}[b]{0.45\columnwidth}
        \includegraphics[width=\textwidth]{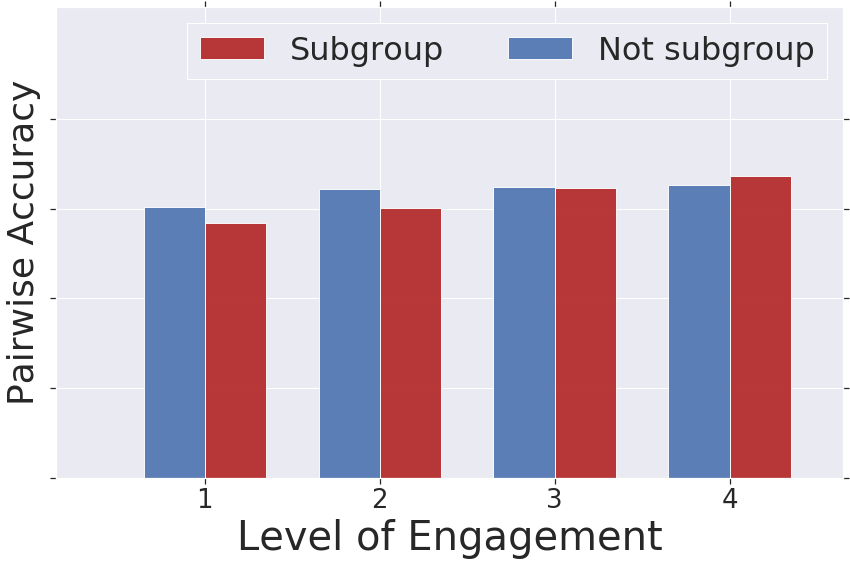}
        \caption{After \pairregshort}
    \end{subfigure}
    \caption{We find some gaps in \emph{overall pairwise accuracy} that are improved through the pairwise regularization.}
    \label{fig:fairness_overall}
\end{figure}

\begin{figure}[t]
    \centering
    \begin{subfigure}[b]{0.45\columnwidth}
        \includegraphics[width=\textwidth]{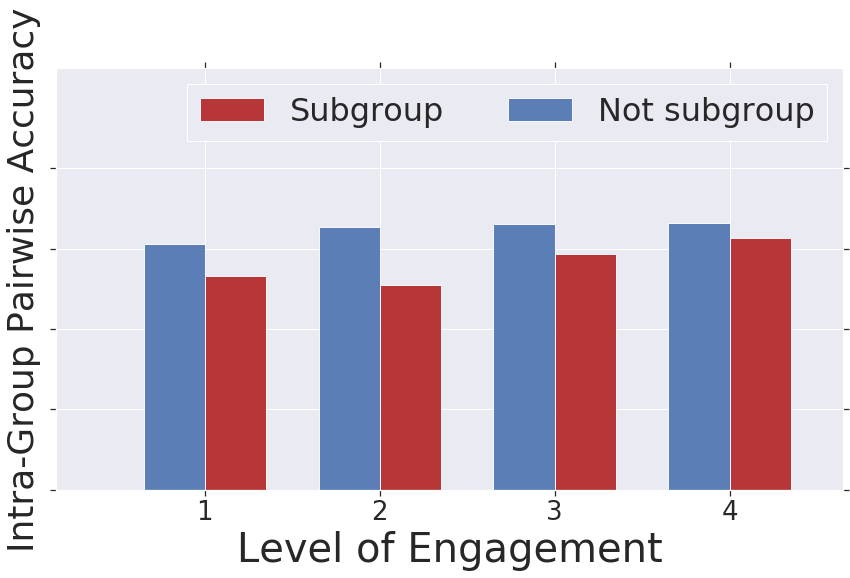}
        \caption{Original}
        \label{fig:prod_fairness_intra}
    \end{subfigure}
    \hspace{2mm}
    \begin{subfigure}[b]{0.45\columnwidth}
        \includegraphics[width=\textwidth]{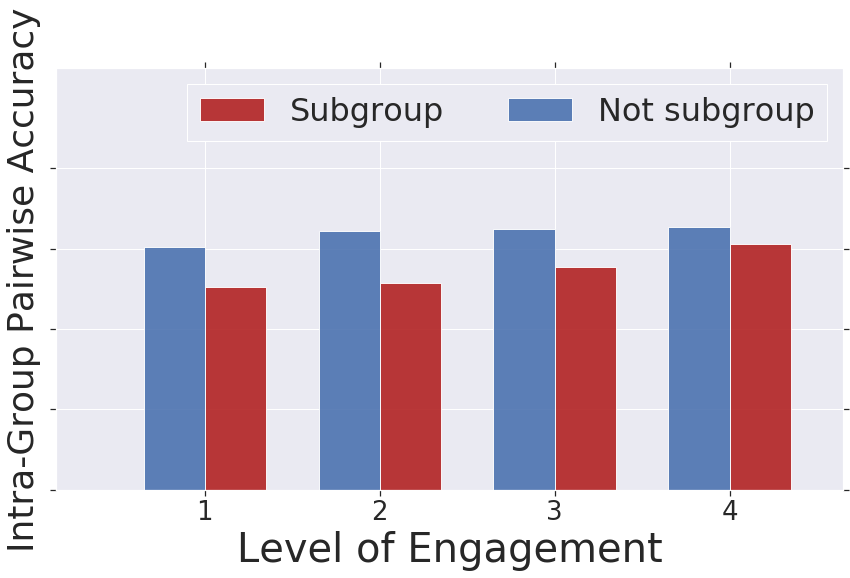}
        \caption{After \pairregshort}
    \end{subfigure}
    \caption{We observe slight differences in \emph{intra-group pairwise accuracy}.}
    \label{fig:fairness_intra}
\end{figure}

\begin{figure*}[th!]
    \centering
    \begin{subfigure}[b]{0.28\textwidth}
        \includegraphics[width=\textwidth]{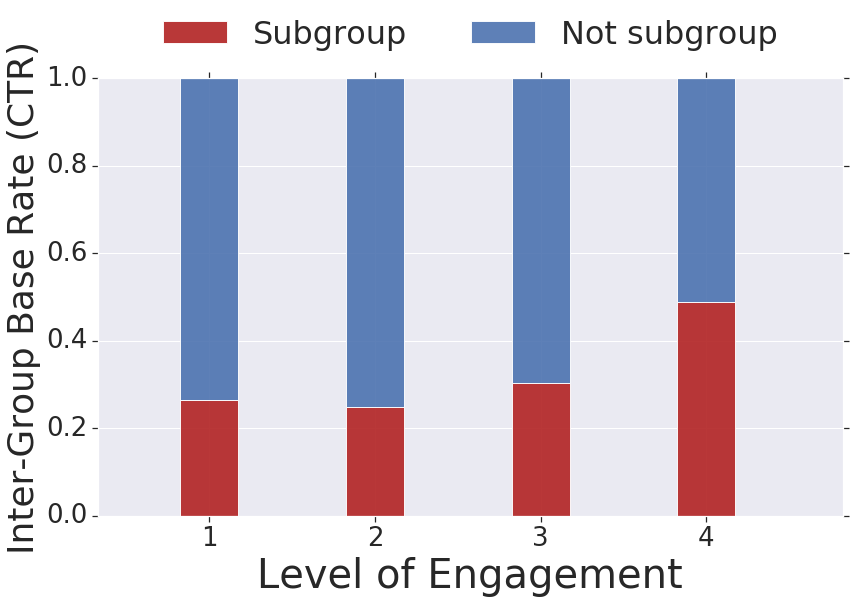}
        \caption{User Preferences}
        \label{fig:ctr}
    \end{subfigure}
    \hspace{2mm}
    \begin{subfigure}[b]{0.28\textwidth}
        \includegraphics[width=\textwidth]{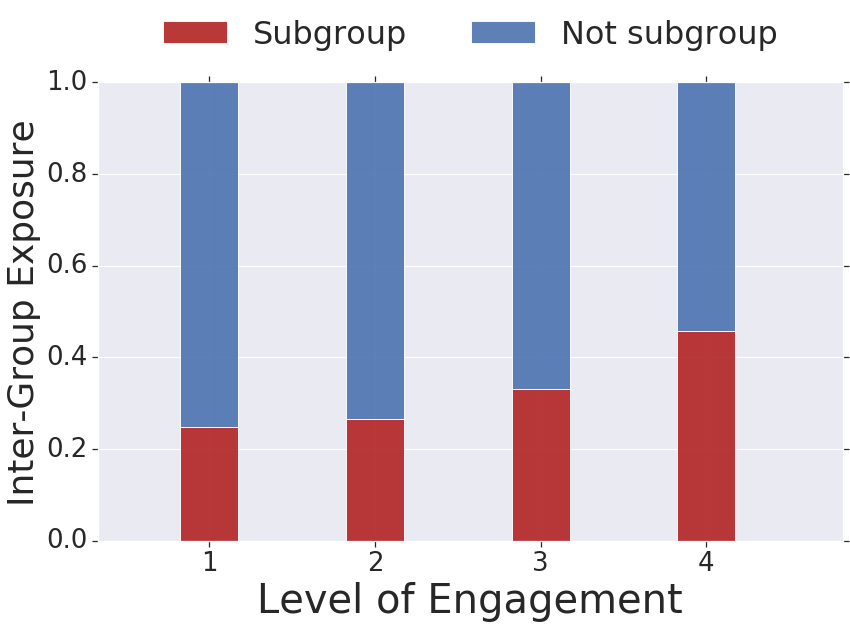}
        \caption{Original Model Exposure}
        \label{fig:prod_exposure}
    \end{subfigure}
    \hspace{2mm}
    \begin{subfigure}[b]{0.28\textwidth}
        \includegraphics[width=\textwidth]{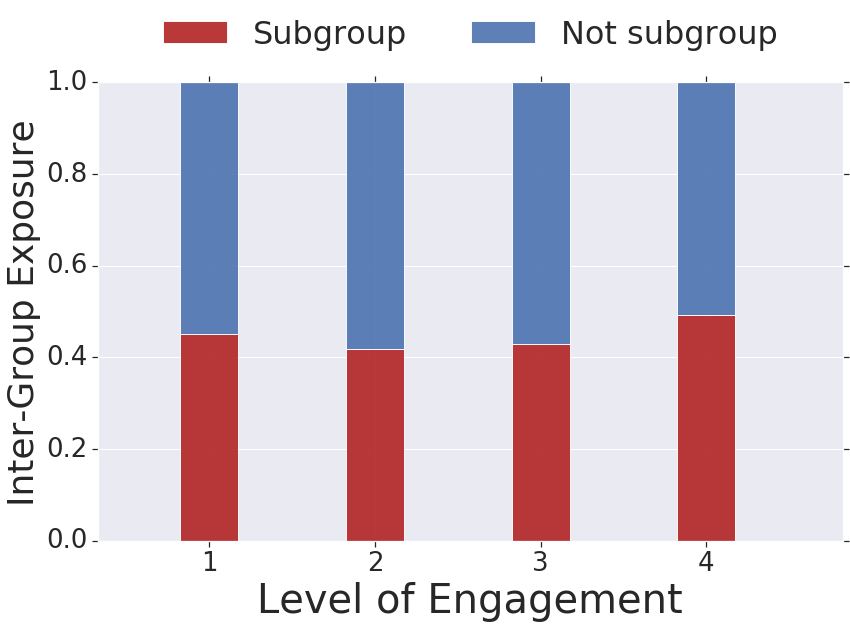}
        \caption{Exposure After Pairwise Reg.}
        \label{fig:test_exposure}
    \end{subfigure}
    \caption{We find that the original model's exposure closely matched the observed user preferences in the data.  In correcting for pairwise fairness we observe that we comparably show more items from the subgroup.}
    \label{fig:exposure}
\end{figure*}

\subsection{Fairness Improvements}
As described above, we apply the \pairreg from Section \ref{sec:modeling} over inter-group pairs of examples so as to optimize for inter-group pairwise fairness.

We see in Figure \ref{fig:fairness_inter} the effect of \pairreg on the inter-group pairwise fairness, the metric that it is most aligned with.  While the regularization decreases the pairwise accuracy of the non-subgroup items, it effectively closes the gap in the inter-group pairwise fairness metric, resulting in only a 2.6\% advantage for non-subgroup item in inter-group pairwise fairness, down from 35.6\%.  Further, while the decrease in pairwise accuracy for the non-subgroup items may appear discouraging, the pairwise accuracy for the non-subgroup in the test model is approximately on-par with the pairwise accuracy metrics we see in intra-group comparisons for non-subgroup items, suggesting the model is no longer taking advantage of the difference in items.

While not our immediate goal, we also examine how improving the inter-group fairness effects the overall pairwise fairness.  As we see in Figure \ref{fig:fairness_overall}, there is a visible improvement in the pairwise accuracy for subgroup items and the gap between the groups is largely closed.  Quantitatively we observe that the relative benefit to the non-subgroup items decreases to 2.5\%, down from 8.3\%.  Intra-group accuracy is not optimized by our \pairreg configuration, and as expected we see little change in intra-group accuracy (Figure \ref{fig:fairness_intra}, with a 16.7\% advantage for non-group items).

Interestingly, in most of our live experiments using models trained with \pairreg we found overall engagement metrics were neutral relative to the production system.  Given the subgroup is a tiny fraction of the overall system, it is reassuring to see that the above fairness benefits do not come at a cost to overall performance. 

Together, this shows that \pairreg is effective in improving the fairness properties of the ranker.

\subsection{How are improvements achieved?}
While the results are compelling, we do further analysis to understand how the regularization is able to close fairness gaps.   To do this, we examine the exposure of items from each group compared to the user preferences, similar in principle to a coarse pairwise calibration analysis.

To understand the user preferences, we measure the percentage of  inter-group pairs for which users prefer (click on) the subgroup item versus the non-subgroup item.  This presents a base rate click-through-rate (CTR) for each group, similar to the analysis in Section \ref{sub:theory_relations}.  As we see in Figure \ref{fig:ctr}, across nearly all levels of engagement, the subgroup items are less likely to be clicked when juxtaposed with a non-subgroup item; interestingly, high-engagement interactions show a nearly even balance of likelihood of a click across the groups.

To understand how the model performs compared to this base CTR, we measure \emph{exposure}: the probability of the model ranking one group's item above that of the other group, irrespective of the user preference\footnote{This is a slight modification of exposure as defined by \citet{DBLP:conf/kdd/SinghJ18}, using a probabilistic form over inter-group comparisons.}.  To be precise:
\begin{align*}
{\rm Exposure}_{\tilde{s},\tilde{z}} = P(\compare | s_j = \tilde{s}, s_{j'} = 1 - \tilde{s}, z_{\vec{q},j'} = \tilde{z})
\end{align*}
As we see in Figure \ref{fig:prod_exposure}, the production model  exposes each group at approximately the same rate as the group's base CTR.

As we can see in Figure \ref{fig:test_exposure}, the exposure of items from each group changes significantly when the model is trained with the \pairreg.  Even with lower levels of engagement, items from the subgroup are ranked higher at significantly higher rate than the base CTR.  This suggests that the regularizer has the effect of showing subgroup items at a higher rate than is natural so as to make sure users interested in subgroup items are recommended them.  This aligns with Lemma \ref{lemma:calibration} suggesting a general tension between calibration and pairwise fairness.
We believe further research on this relationship and more generally how to improve model accuracy can help alleviate this tension, but for the time-being find this to be a reasonable trade-off.

\section{Conclusion}
In this work we have provided a tractable way to get unbiased measurements of recommender system ranking fairness.  We are able to do this through pairwise experiments to observe user preferences.  Based on this experimental data, we can evaluate and decompose recommender system fairness to see if a model systematically mis-ranks or under-ranks items from a particular group.  We show that this measure aligns with ranking fairness definitions but is not covered by pointwise fairness measures.  We ultimately offer a novel \pairreg approach to improve recommender system fairness during training, and show that it significantly improves fairness metrics in a large-scale production system.

\vspace{2mm}
{\footnotesize 
\noindent \textbf{Acknowledgements:}
The authors would like to thank Ben Packer, Xuezhi Wang, and Andrew Cotter for their helpful comments during the preparation of this paper.
}

\cleardoublepage

\bibliography{alex}

\end{document}